\documentclass[reqno,11pt]{amsart}
\usepackage[left=3.5cm, top=3.2cm, bottom=3.2cm, right=3.5cm]{geometry}
\usepackage{amssymb}
\usepackage{amsxtra}
\usepackage{amsmath}
\usepackage{amsfonts}
\usepackage{color}

\numberwithin{equation}{section}

\newtheorem{thm}{Theorem}[section]

\newtheorem{prop}[thm]{Proposition}

\newtheorem{lem}[thm]{Lemma}
\newtheorem{de}[thm]{Definition}
\newtheorem{rem}[thm]{Remark}
\newtheorem{ex}[thm]{Example}

\newcommand{\eqa}{\begin{eqnarray}}
\newcommand{\eeqa}{\end{eqnarray}}
\newcommand{\beq}{\begin{equation}}
\newcommand{\eeq}{\end{equation}}

\begin{document}
\title[]{Generalized Legendre transformations and symmetries of the WDVV equations}
\author[]{
{Ian A.B. Strachan \ \  Richard Stedman}}

\address{School of Mathematics and Statistics, University of Glasgow, Glasgow G12 8QW, United Kingdom}

\email{\mbox{ian.strachan@glasgow.ac.uk, \quad richard.stedman@gmail.com}}
\date{}


\dedicatory{}
\date{\today}

\begin{abstract}
The Witten-Dijkgraaf-Verlinde-Verlinde (or WDVV) equations, as one would expect from an integrable system, has many symmetries, both continuous and discrete. One class - the so-called Legendre transformations - were introduced by Dubrovin. They are a discrete set of symmetries between the stronger concept of a Frobenius manifold, and are generated by certain flat vector fields. In this paper this construction is generalized to the case where the vector field (called here the Legendre field) is non-flat but satisfies a certain set of defining equations. One application of this more general theory is to generate the induced symmetry between almost-dual Frobenius manifolds whose underlying Frobenius manifolds are related by a Legendre transformation. This also provides a map between rational and trigonometric solutions of the WDVV equations.
\end{abstract}
\maketitle

\tableofcontents
\newpage

\section{Introduction}

The Witten-Dijkgraaf-Verlinde-Verlinde (or WDVV) equations of associativity
\[
\frac{\partial^3F(t)}{\partial t^\alpha\partial t^\beta\partial t^\nu}\eta^{\nu\mu}
\frac{\partial^3F(t)}{\partial t^\mu\partial t^\gamma\partial t^\lambda}=
\frac{\partial^3F(t)}{\partial t^\beta\partial t^\gamma\partial t^\nu}\eta^{\nu\mu}
\frac{\partial^3F(t)}{\partial t^\alpha\partial t^\mu\partial t^\lambda},
\]
have been much studied from a variety of different points of view, amongst them, topological quantum field theories, Seiberg-Witten theory, singularity theory and integrable systems. Geometrically, a solution defines a multiplication $\circ\,:  T\mathcal{M} \times T\mathcal{M}\rightarrow  T\mathcal{M}$ of vector fields, i.e.
\[
\frac{\partial~}{\partial t^\alpha} \circ \frac{\partial~}{\partial t^\beta} = \eta^{\mu\nu}\frac{\partial^3 F}{\partial t^\alpha \partial t^\beta \partial t^\mu} \,\frac{\partial~}{\partial t^\nu}\,.
\]
The metric $\eta$, used to raise and lower indices, is flat and the coordinates $\{ t^\alpha \}$ are flat coordinates, i.e. the components of the metric is this coordinate system are constants.

 A symmetry of the WDVV equations are transformations
\begin{itemize}
\item[(i)] $t^\alpha \mapsto {\tilde t}^\alpha\,,$
\item[(ii)] $\eta_{\alpha\beta} \mapsto {\tilde\eta}_{\alpha\beta}\,;$
\item[(iii)] $F \mapsto {\tilde F}$
\end{itemize}
that preserve the WDVV equations, i.e. which map solutions of the WDVV equations to new solutions. In \cite{Du1} the following transformation - called a transformation of Legendre type - was shown to be such a symmetry:
\begin{eqnarray*}
{\tilde t}_\alpha & = & \displaystyle{ \frac{\partial^2 F}{\partial t^\alpha \partial t^\kappa}}\,,\\
\displaystyle{ \frac{\partial^2 {\tilde F}}{\partial {\tilde t}^\alpha \partial {\tilde t}^\beta}} &=& \displaystyle{ \frac{\partial^2 {F}}{\partial {t}^\alpha \partial {t}^\beta}}\,,\\
{\tilde\eta}_{\alpha\beta} &= & \eta_{\alpha\beta}\,.
\end{eqnarray*}
Here $\kappa$ is a distinguished variable which labels the Legendre transformation. Moreover, this is more than just a map between solutions of the WDVV equations: it preserves the stronger notion of a Frobenius manifold.

This transformation may be reinterpreted more geometrically in terms of a  vector field
\[
\partial = \frac{\partial~}{\partial t^\kappa}\,, \qquad\qquad \kappa\in\{1\,,\ldots\,,{\rm dim}\mathcal{M} \}\,.
\]
With this the metric transformation may be written
\[
{\tilde\eta}(X,Y) = \eta(\partial\circ X,\partial\circ Y)
\]
for arbitrary vector fields $X$ and $Y\,.$ The field $\partial$ is distinguished by the fact that it is flat for the Levi-Civita connection of the metric $\eta\,,$ namely $\nabla_X\partial=0\,.$

The aim of this paper is to study more general Legendre-type transformations where the vector field $\partial$ is replaced by a field satisfying the condition
\begin{equation}\label{Legendre}
X\circ \nabla_Y \partial = Y \circ \nabla_X \partial\,,\qquad\qquad X\,,Y\in T\mathcal{M}\,.
\end{equation}
Such fields will be called Legendre fields. They have appeared in the literature before: in \cite{LPR} they are used as generators of sets of commuting systems of hydrodynamic type. Here we follow their role as a map between torsion free metric connections \cite{DS}. These fields are defined and studied in Section 2. With such fields more general transformations between solutions of WDVV equations (but not for the stronger notion of a transformation between Frobenius manifold) may be constructed, and this construction is given in Section 4, after a submanifold theory of Legendre fields is developed in Section 3. Finally, in Section 4, the theory is applied to the induced transformation between almost-dual Frobenius manifolds. In particular, it is shown how the twisted-Legendre transformation introduced in \cite{RS} fits into this more general picture.

Some of the results in Section 2 have appeared in \cite{DS} and also in \cite{LPR}. In the former, the emphasis was on Legendre transformation as maps between connections. While this is how they are initially defined, Theorem \ref{symmetry} gives a coordinate description of the transformation between flat coordinate systems, and is a direct generalization of the original transformation as defined in \cite{Du1}. Similarly, some of the results in Section 5 follows from the constructions in \cite{DS} as a special case, when an arbitrary eventual identity is replaced by a specific example - the Euler vector field. However in this special case one may simplify the proofs and identify those Legendre fields that are actually flat. The link with the work in \cite{LPR} is outlined in the final section.

\section{Legendre fields and their properties}

Consider the following structure: $ \{ \mathcal{M}, \eta,\nabla,\circ\}$, where $\mathcal{M}$ is a manifold and
\begin{itemize}
\item[(a)] $\eta$ is a metric;
\item[(b)] $\nabla$ a connection (not, at this stage, metric or torsion free);
\item[(c)] $\circ: T\mathcal{M} \times T\mathcal{M} \rightarrow T\mathcal{M}$ is a commutative, associative multiplication (so a completely symmetric $(2,1)$-tensor)
with a unit $e$ which is compatible with the metric, i.e.
\[
\eta(X \circ Y,Z) = \eta(X, Y\circ Z)\,;
\]
\item[(d)] $\nabla\circ$ is a completely symmetric $(3,1)$-tensor.
\end{itemize}
A Legendre field will map metrics and connections to new metric and connections which preserve certain basic properties.
\begin{de}
Let $\partial$ be an invertible vector field (so $\partial^{-1} \circ \partial = e$). Let
\begin{eqnarray*}
{\tilde\eta}(X,Y) & = & \eta( \partial\circ X,\partial\circ Y)\,,\\
{\tilde\nabla}_XY&=& \partial^{-1} \circ \nabla_X(\partial\circ Y)\,.
\end{eqnarray*}
\end{de}
\noindent The following follows from direct calculation.
\begin{lem}
For all vector fields $X\,,Y\,,Z \in T\mathcal{M}$
\begin{eqnarray*}
( {\tilde\nabla}_X\eta)(Y,Z) & = & \left( \nabla_X\eta\right)(\partial\circ X,\partial\circ Y)\,,\\
T^{\tilde\nabla}(X,Y) & = & T^\nabla(X,Y) + \partial^{-1} \circ \left\{ Y \circ \nabla_X \partial - X \circ \nabla_Y \partial\right\}\,,\\
R^{\tilde\nabla}(X,Y)(Z) & = & \partial^{-1} \circ R^\nabla(X,Y) (\partial\circ Z)\,,
\end{eqnarray*}
where $T^\nabla$ and $R^\nabla$ are the torsion and curvature tensors of the connection $\nabla\,.$
\end{lem}
\noindent The torsion result uses the symmetry of $\nabla\circ$, i.e. $(\nabla_X\circ)(\partial,Y) = (\nabla_Y\circ)(\partial,X)\,.$ This result motivates the following:
\begin{de}
A Legendre field $\partial \in T\mathcal{M}$ is a solution to the equation
\[
X \circ \nabla_Y \partial = Y \circ \nabla_X \partial
\]
for all vectors $X\,,Y\in T\mathcal{M}\,.$
\end{de}
Thus a Legendre field will map torsion free connections to torsion free connections, metric connections to metric connections, and zero-curvature conditions to zero-curvature conditions.
The definition appears to be overdetermined, with more equations than unknowns. However if one sets $Y=e$ one obtains the manifestly determined system
\[
\nabla_X\partial = X \circ \nabla_e\partial\,.
\]
Conversely, a solution of this equation implies that $\partial$ is a Legendre field, so this provides and alternative definition of a Legendre field.

In what follows we will apply Legendre fields to map solutions of the WDVV equations to new solutions, but the ideas may be applied more generally to situations where one has torsion free connections and non-metric connections, or curvature, such as Riemannian $F$-manifolds \cite{AL,LPR}.

Before doing this we derive certain basic properties of such Legendre fields. From now on we assume that $\nabla$ is the torsion free metric connection for $\eta\,.$ Clearly sums and constant scalar-multiplies of Legendre fields are again Legendre fields. In addition, when suitably defined, inverses and composition of Legendre fields are again Legendre fields.

\begin{prop}
If $\partial$ is a Legendre field for $\nabla$ then $\partial^{-1}$ is a Legendre field for $\tilde\nabla\,.$
\end{prop}
\begin{proof}
Expanding the symmetry condition $(\nabla_X\circ)(e,e)=(\nabla_e\circ)(X,e)$ yields
\[
\nabla_Xe = X\circ \nabla_e e
\]
and hence $e$ is a Legendre field (without the need for a condition such as $\nabla e=0$).
With this,
\begin{eqnarray*}
X\circ {\tilde \nabla}_Y \partial^{-1} & = & X \circ \partial^{-1} \circ \nabla_Y(\partial\circ\partial^{-1})\,,\\
& = & \partial^{-1}\circ X \circ \nabla_Y e\,.
\end{eqnarray*}
But since $e$ is a Legendre field the right-hand expression is symmetric in $X$ and $Y\,.$ Hence the result.
\end{proof}

\noindent Note that $\tilde\nabla_X \partial^{-1} = \partial^{-1} \circ \nabla_X e$ so even if $\partial$ is flat for $\nabla$ then it does not follow that $\partial^{-1}$ is flat for $\tilde\nabla\,,$ a further condition is required, i.e. $\nabla e=0\,.$

\begin{prop}
Consider connections $\nabla^{(i)}\,,i=1,2,3$ with interconnecting Legendre fields $\partial^i_j\,,$
\[
\nabla^{(1)} \xrightarrow{~~~\partial^2_1~~~} \nabla^{(2)}\xrightarrow{~~~\partial^3_2~~~} \nabla^{(3)}\,.
\]
Then
\[
\nabla^{(1)} \xrightarrow{~~~\partial^3_1~~~} \nabla^{(3)}
\]
where $\partial^3_1 = \partial^2_1 \circ \partial^3_2\,.$
\end{prop}

\begin{proof}
Since
\begin{eqnarray*}
\nabla^{(2)}_XY & = & (\partial^2_1)^{-1} \circ \nabla^{(1)}_X (\partial^2_1\circ Y) \quad\quad{\rm where~} X\circ\nabla^{(1)}_Y \partial^2_1 =Y\circ\nabla^{(1)}_X \partial^2_1 \,,\\
\nabla^{(3)}_XY & = & (\partial^3_2)^{-1} \circ \nabla^{(2)}_X (\partial^3_2\circ Y) \quad\quad{\rm where~} X\circ\nabla^{(2)}_Y \partial^3_2 =Y\circ\nabla^{(2)}_X \partial^3_2 \,,
\end{eqnarray*}
it follows that
\[
\nabla_X^{(3)}Y=(\partial^2_1\circ\partial^3_2)^{-1} \circ \nabla_X^{(1)} \left[ (\partial^2_1\circ\partial^3_2)\circ Y\right]
\]
and that
\begin{eqnarray*}
X\circ\nabla^{(1)}_Y ( \partial^2_1\circ \partial^3_2) & = & X \circ \partial^2_1 \circ \nabla^{(2)}_Y \partial^3_2\,,\\
&=& Y \circ \partial^2_1 \circ \nabla^{(2)}_X \partial^3_2\,,\\
&=& Y\circ\nabla^{(1)}_X(\partial^2_1\circ \partial^3_2)\,.
\end{eqnarray*}
Hence $\partial^3_1=\partial^2_1\circ\partial^3_2$ is a Legendre field between the connections $\nabla^{(1)}$ and $\nabla^{(3)}\,.$
\end{proof}
New Legendre fields may also be constructed from old.
\begin{prop}
Let $\partial$ be a Legendre field. Then:
\begin{itemize}
\item[(a)] if $\nabla e=0$ and $R^\nabla(X,e)=0$ then $\nabla_e\partial$ is also a Legendre field;
\item[(b)] the solution $\partial_{\rm new}$ of the equation
\[
\nabla_X \partial_{\rm new} = X \circ \partial
\]
is a Legendre field;
\end{itemize}
\end{prop}
\begin{proof}
\begin{itemize}
\item[(a)] From the easily proved identity
\begin{eqnarray*}
\nabla_X(\nabla_e\partial) - X \circ \nabla_e(\nabla_e\partial) & = & R(X,e)\partial\,,\\
& = & 0 \,
\end{eqnarray*}
it follows that $X \circ \nabla_Y \left( \nabla_e \partial \right)$ is symmetric in $X$ and $Y\,.$
\item[(b)] Since $Y \circ \nabla_X \partial_{\rm new} = X\circ Y \circ \partial$ is symmetric in $X$ and $Y$ the result follows.
Note, this does not use the property that $\partial$ is a Legendre field. However, from part (a) with $X=e$ it follows that $\nabla_e \partial_{\rm new} = \partial$ so if the conditions in part (a) hold the vector field $\partial$ must be Legendre.
\end{itemize}

\end{proof}

\subsection{Homogeneous structures}

If the original multiplication is homogeneous, i.e. there exists an Euler vector $E$ such that $\mathcal{L}_E\circ=\circ\,,$ with the metric $\eta$ satisfying the homogeneity condition $\mathcal{L}_E\eta= (2-d)\, \eta\,,$ then such conditions are preserved if the Legendre field is similarly homogeneous. The following result is straightforward and will be given without proof.

\begin{lem} Let $\mathcal{L}_E\circ=\circ\,,$  and suppose
\begin{eqnarray*}
\mathcal{L}_E \eta & = & (2-d) \,\eta\,,\\
\mathcal{L}_E \partial & = & \mu \,\partial
\end{eqnarray*}
for constants $d$ and $\mu\,.$ Then
\[
\mathcal{L}_E {\tilde\eta} = (4-d + 2 \mu) \,{\tilde\eta}\,.
\]
\end{lem}
\noindent Similar results may be proved for the homogeneity of $\nabla_e\partial$ and related objects. Before applying these ideas to generate new solutions of the WDVV equations from old, we develop a submanifold theory for Legendre fields.

\section{A submanifold theory for Legendre fields}
Consider a submanifold $\mathcal{N} \subset \mathcal{M}$ with an orthogonal decomposition, with respect to the metric $\eta\,,$ of the tangent bundle $T\mathcal{M} = T\mathcal{N} \oplus T\mathcal{N}^\perp\,.$ We asume that $\mathcal{N}$ is a natural submanifold \cite{S}, i.e. that $T\mathcal{N}\circ T\mathcal{N} \subset T\mathcal{N}\,.$ In order to study the restriction of a Legendre field to a submanifold we first fix some notation.
\begin{eqnarray*}
\nabla_XY & = & {\overline\nabla}_XY + \alpha(X,Y)\,,\qquad\qquad X,Y \in T\mathcal{N}\,,\\
\nabla_x\xi & = & - A_\xi(X) + {\overline\nabla}^\perp_X\xi\,,\qquad\qquad \xi\in T\mathcal{N}^\perp\,,
\end{eqnarray*}
where the first terms on the right-hand side lie in $T\mathcal{N}$ and the second in $T\mathcal{N}^\perp\,.$ Standard ideas from submanifold theory imply that
$\alpha$ is symmetric and that $\eta(A_\xi(X),Y)=\eta(\alpha(X,Y),\xi)\,.$

A Legendre field $\partial$ on $\mathcal{M}$ when restricted to $\mathcal{N}$ decomposes as
\[
\partial = \partial_{||}+\partial_\perp\,,\qquad\qquad \partial_{||} \in T\mathcal{N}\,, \partial_\perp\in T\mathcal{N}^\perp\,,
\]
and on decomposing the Legendre condition (\ref{Legendre}) the component in $T\mathcal{N}$ yields the equation
\[
Y\circ {\overline\nabla}_X \partial_{||} - X\circ {\overline\nabla}_Y \partial_{||} = A_{\partial_\perp}(X)\circ Y - A_{\partial_\perp}(Y)\circ X\,,\qquad X,Y \in T\mathcal{N}\,.
\]
This uses the result $\xi\circ X \in T{\mathcal{N}}^\perp$ which follows from  the natural submanifold condition $0=\eta(X\circ Y,\xi)=\eta(Y,X\circ\xi)$\,.
Hence
\[
\eta\left(
Y\circ {\overline\nabla}_X \partial_{||} - X\circ {\overline\nabla}_Y \partial_{||} ,Z\right) =
\eta(\alpha(X,Y\circ Z) - \alpha(Y,X\circ Z), \partial_\perp)\,.
\]
Similarly decomposing the symmetry condition $(\nabla_X\circ)(Y,Z) = (\nabla_Y\circ)(X,Z)$ for $X\,,Y\,,Z\in T\mathcal{N}$ yields both the total symmetry of ${\overline\nabla}\circ$ and
\[
\alpha(X,Y\circ Z) - \alpha(Y,X\circ Z) = Y \circ \alpha(X,Z) - X \circ \alpha(Y,Z)\,.
\]
Thus
\[
\eta\left(
Y\circ {\overline\nabla}_X \partial_{||} - X\circ {\overline\nabla}_Y \partial_{||} ,Z\right) =
\eta\left(
Y, \alpha(X,Z)\circ \partial_\perp\right)-
\eta\left(
X, \alpha(Y,Z)\circ \partial_\perp\right)\,.
\]

Hence we arrive at the obstruction for a parallel component of a Legendre field on a natural submanifold to be Legendre.
\begin{lem}
The vector field $\partial_{||}$ is a Legendre field if and only the tensor
\[
\eta\left(
X, \alpha(Y,Z)\circ \partial_\perp\right)
\]
is totally symmetric.
\end{lem}
Thus what controls the properties of the nascent induced Legendre field is the nature of $T\mathcal{N}^\perp \circ T\mathcal{N}^\perp\,.$ If
$T\mathcal{N}^\perp \circ T\mathcal{N}^\perp\subset T\mathcal{N}^\perp$ the above tensor is trivially totally symmetric - it vanishes identically. This happens in the case of a discriminant submanifold of a semi-simple Frobenius manifold \cite{S}.

\begin{ex}
Suppose the multiplication on $\mathcal{M}$ is semi-simple. Then the component $\partial_{||}$ of a Legendre field on a discriminant submanifold is a Legendre field for the induced structures on $\mathcal{N}\,.$
\end{ex}

\begin{ex}
From above, the unity field $e$ is a Legendre field. On decomposing this on a natural submanifold, $e=e_{||}+e_\perp$ gives
\[
\xi \circ (e_{||}+e_\perp) = \xi\,,\qquad\qquad \xi \in T\mathcal{N}^\perp\,.
\]
With this
\[
\eta(\xi \circ e_{||},X) + \eta(\xi\circ e_\perp,X) = \eta(\xi,X) =0\,.
\]
But $\eta(\xi \circ e_{||},X) = \eta( \xi, e_{||} \circ X) =0$ since $\mathcal{N}$ is a natural submanifold. Thus $\xi\circ e_\perp\in T\mathcal{N}^\perp$ for all $\xi\in T\mathcal{N}^\perp\,.$
Thus $\eta(\alpha(X,Z)\circ e_{\perp},Y)=0$ and hence $e_{||}$ is a Legendre field on the submanifold.
\end{ex}

\section{Symmetries of the WDVV equations}

The WDVV equations may be described geometrically by the condition that the deformed connection
\[
{}^{(\lambda)}\nabla_XY=\nabla_XY - \lambda X\circ Y
\]
has vanishing curvature \cite{Du1}. In particular $R^\nabla=0$ (the $\lambda^0$-term in $R^{ {}^{(\lambda)}\nabla }$) and hence there exists a coordinate system $\{t^\alpha\}$ in which the components $\eta(\frac{\partial~}{\partial t^\alpha},\frac{\partial~}{\partial^\beta})$ are constant. Given a Legendre field it follows from above that if ${}^{(\lambda)}{\tilde\nabla}_XY={\tilde\nabla}_XY - \lambda X\circ Y$ then
\[
    R^{ {}^{(\lambda)}{\tilde\nabla}} (X,Y)Z = \partial^{-1} \circ R^{   {}^{(\lambda)}  \nabla} (X,Y)(\partial\circ Z)
\]
and hence that the new deformed connection ${}^{(\lambda)}{\tilde\nabla}_XY$ is also flat and hence defines a new - Legendre transformed - solution of the WDVV equations. It is important to note that geometrically the multiplication does not change when such a transformation is applied. However, the coordinate expression for this tensor does depend on the choice of flat-coordinate system and these do change under such a transformation.

The following Theorem gives such coordinate dependent formulae for the new structures in terms of the old, and is a direct generalisation of the original Legendre transformation present in \cite{Du1}.

\begin{thm}\label{symmetry}
Let $\{t^\alpha\}$ and $\{ {\tilde t}^\alpha\}$ be the flat coordinates of the metrics $\eta$ and $\tilde\eta$ respectively. Then, given a solution of the WDVV equations with prepotential $F$ and a  Legendre field $\partial$, new flat coordinates, metric coefficients and prepotential are given by:
\begin{itemize}
\item[(a)] $\displaystyle{{\tilde\eta}_{\alpha\beta} = \eta_{\alpha\beta}}\,;$
\vskip 2mm
\item[(b)] $ \displaystyle{\frac{\partial {\tilde t}^\alpha }{\partial t^\beta} = \partial^\sigma c^\alpha_{\sigma\beta}}\,;$
\vskip 2mm
\item[(c)] ${\tilde F}_{\alpha\beta} = F_{\alpha\beta}\,.$
\end{itemize}
\end{thm}

\begin{proof}

\begin{itemize}
\item[(a)] Since $\tilde\eta(X,Y) =\eta(\partial\circ X,\partial \circ Y)$ it follows that, up to an inessential linear transformation, that
\begin{equation}
\frac{\partial~}{\partial t^\alpha} = \partial \circ \frac{\partial~}{\partial {\tilde t}^\alpha}
\label{CoV}
\end{equation}
\noindent and hence that
\begin{eqnarray*}
{\tilde\eta}_{\alpha\beta} & = & {\tilde\eta}\left( \frac{\partial~}{\partial {\tilde t}^\alpha},\frac{\partial~}{\partial {\tilde t}^\beta}\right)\,,\\
&=&  {\eta}\left( \frac{\partial~}{\partial {t}^\alpha},\frac{\partial~}{\partial {t}^\beta}\right)\,,\\
&=& \eta_{\alpha\beta}\,.
\end{eqnarray*}

\item[(b)] From (\ref{CoV}) it follows that
\begin{eqnarray*}
\frac{\partial~}{\partial t^\alpha} & = & \frac{\partial t^\beta}{\partial {\tilde t}^\alpha} \left( \partial \circ \frac{\partial~}{\partial {t}^\beta}\right)\,,\\
&=&\frac{\partial t^\beta}{\partial {\tilde t}^\alpha} \, \partial^\sigma c_{\sigma\beta}^\mu \frac{\partial~}{\partial t^\mu}\,.\\
\end{eqnarray*}
Hence
\[
\frac{\partial t^\beta}{\partial {\tilde t}^\alpha} \, \partial^\sigma c_{\sigma\beta}^\mu = \delta^\mu_\alpha
\]
and (b) follows. On lowering an index and using the symmetry of $c_{\sigma\alpha\beta}$ it follows that
\[
\frac{\partial \tilde{t}_\alpha}{\partial t^\beta} = \frac{\partial \tilde{t}_\beta}{\partial t^\alpha}
\]
and hence $\tilde{t}_\alpha = \partial_{t^\alpha} h$ for some locally defined function $h\,.$

\item[(c)] To prove the final formula,
\begin{eqnarray*}
\frac{\partial {\tilde F}_{\alpha\beta}}{\partial {\tilde t}^\gamma} = {\tilde c} \left(
\frac{\partial~}{\partial {\tilde t}^\alpha},
\frac{\partial~}{\partial {\tilde t}^\beta},
\frac{\partial~}{\partial {\tilde t}^\gamma}\right) & = & {\tilde{\eta}}
\left(\frac{\partial~}{\partial {\tilde t}^\alpha}\circ\frac{\partial~}{\partial {\tilde t}^\beta},\frac{\partial~}{\partial {\tilde t}^\gamma}\right)\,,\\
&=&
\eta\left(
\frac{\partial~}{\partial {t}^\alpha}\circ\frac{\partial~}{\partial {t}^\beta},\frac{\partial~}{\partial {t}^\sigma}\right) \frac{\partial t^\sigma}{\partial {\tilde t}^\gamma}\,,\\
&=&
\frac{\partial t^\sigma}{\partial {\tilde t}^\gamma}
c\left(
\frac{\partial~}{\partial {t}^\alpha},
\frac{\partial~}{\partial {t}^\beta},
\frac{\partial~}{\partial {t}^\sigma}\right)
= \frac{\partial F_{\alpha\beta}}{\partial {\tilde t}^\gamma} \,.
\end{eqnarray*}
Hence, on integrating, up to an inessential constant, part (c) follows.
\end{itemize}
\end{proof}

The geometry of the deformed flat connection encodes a canonical class of Legendre fields. Expanding a flat section ${}^{(\lambda)}\nabla_X s=0$ as a power series
\[
s=\sum_{n=0}^\infty \lambda^n \partial_{(n)}
\]
and equating coefficients gives
\begin{eqnarray}
\nabla_X \partial_{(0)} & = & 0 \,, \label{zeroth} \\
\nabla_X \partial_{(n)} & = & X \circ \partial_{(n-1)}\,.\label{recursion}
\end{eqnarray}
Thus each of the fields $\partial_{(n)}$ are Legendre fields. Conversely, starting from a flat vector field $\partial_{(0)}$ one may recursively construction the flat section, with each $\partial_{(n)}$ being a Legendre field. If $\partial_{(0)} = \frac{\partial~}{\partial t^\kappa}$ for some $\kappa$ one obtains an infinite family of Legendre fields labeled by $(n,\kappa)\,.$

Note that, when written in coordinate form (\ref{recursion}) takes the form
\[
\frac{\partial~}{\partial t^\alpha} \partial^\beta_{(n,\kappa)} = c^\beta_{\alpha\sigma} \partial^\sigma_{(n-1,\kappa)}\,.
\]
Furthermore, it was shown in \cite{Du1} that the vector field may be written in terms of (scalar) Hamiltonian densities
\[
\partial_{(n,\kappa)} = \eta^{\alpha\beta} \frac{\partial h_{(n,\kappa)}}{\partial t^\alpha} {\frac{\partial~}{\partial t^\beta}}\,.
\]
With this the coordinate transformation takes a simple form:
\begin{eqnarray*}
\frac{\partial{\tilde t}^\alpha}{\partial t^\beta} & = & \partial^\sigma_{(n,\kappa)} c^\alpha_{\sigma\beta}\,,\\
&=& \frac{\partial~} {\partial t^\beta} \partial^\alpha_{(n+1,\kappa)}
\end{eqnarray*}
and hence ${\tilde t}^\alpha = \partial^\alpha_{(n+1,\kappa)}$ or ${\tilde t}_\alpha = \displaystyle{\frac{\partial h_{(n+1,\kappa)}}{\partial t^\alpha}}\,$ (on lowering an index with the metric $\eta$).

\begin{ex}\label{Gen Leg}
Consider the two-dimensional Frobenius manifold defined by the prepotential\footnote{In examples we lower indices for notational convenience.}
\[
F=\frac{1}{2}t_1^2 t_2+e^{t_2},
\]
and Euler vector field
\[
E=t_1\partial_{t_1}+2\partial_{t_2},
\]
On writing $\partial=a(t_1,t_2)\partial_{t_1}+b(t_1,t_2)\partial_{t_2}$ the Legendre field condition implies
\begin{eqnarray*}	
\partial_{t_1} a  & = & \partial_{t_2} b,\\
e^{t_2}\partial_{t_1} b &  = & \partial_{t_2} a.
\end{eqnarray*}
If we additionally impose homogeneity $\mathcal{L}_E\partial=\mu\partial$ we obtain
\begin{eqnarray*}
& \quad& t_1\frac{\partial a}{\partial t_1}+2\frac{\partial a}{\partial t_2}-=(\mu+1) a,\\
& \quad& t_1\frac{\partial b}{\partial t_1}+2\frac{\partial b}{\partial t_2}=\mu b.
\end{eqnarray*}
These imply that $a(t_1,t_2)=t_1^{\mu+1} A(z)$, $b(t_1,t_2)=t_1^\mu B(z)$ where $\displaystyle{z=t_1^{-2} e^{t_2}}$ and with these the Legendre field conditions (\ref{Legendre}) become the following ordinary differential equations for the functions $A$ and $B$:
\begin{eqnarray*}
(\mu+1)A(z)-2zA'(z) & = & zB'(z),\\
\mu B(z)-2zB'(z) & = & A'(z).
\end{eqnarray*}
On eliminating $B(z)$ and setting $w=4z$ we obtain the equation
\[
w(1-w)A''(w)+\frac{2\mu-1}{2}wA'(w)-\frac{\mu(\mu+1)}{4}A(w)=0\,.
\]
This is a hypergeometric equation
\[
w(1-w)A''(w)+[c-(a+b+1)w]A'(w)- a b A(w)=0,
\]
with $ a=-\frac{\mu}{2},b=-\frac{\mu+1}{2}$ and $c=0\,.$ Explicit solutions may hence be found using the well-known general theory of such hypergeometric equations (see \cite{Sted}).
\end{ex}

\section{Twisted Legendre transformation and almost-duality}

Given a Frobenius manifold with metric $\eta$ and multiplication $\circ$ one may construct an associated manifold - an almost-dual Frobenius manifold - with new metric and multiplication
\begin{eqnarray*}
g(X,Y) & = & \eta(E^{-1}\circ X,Y)\,,\\
X\star Y & = & E^{-1}\circ X\circ Y
\end{eqnarray*}
where $E$ is the Euler vector field and $E^{-1}\circ E=e\,.$ It was shown in \cite{Du1,Du2} that $g$ is flat and compatible with the new product and, moreover, that a dual solution of the WDVV equations (written in the flat coordinates of the metric $g$) may be constructed. While the metrics $\eta$ and $g$ are related by the simple formula above, the corresponding connections, $\nabla$ and $\nabla^\star$, are related by a more complicated formula (see, for example, \cite{H})
\[
\nabla^\star_XY = E\circ\nabla(E^{-1}\circ Y) - (\nabla_{E^{-1}}\circ Y)\circ X + \frac{1}{2} (3-d) E^{-1}\circ X\circ Y\,,
\]
where $\mathcal{L}g=(3-d)g$ for some constant $d\,.$ Since we have an almost-dual solution of the WDVV equations, we automatically have a flat pencil of connections
\[
{}^\lambda\nabla^\star_XY = \nabla^\star_XY - \lambda X\star Y
\]
and hence the above theory of generalized Legendre transformations may be applied directly to solutions by replacing $\eta$ with $g$ and $\nabla$ with $\nabla^\star$ and replacing the coordinates $\{ t^\alpha \}$ by the flat coordinates of the metric $g\,.$ However, there is a distinguished Legendre transformation between such almost-dual solutions.

We have, schematically\footnote{Here, and for the rest of this section, $(\eta\,,\nabla\,,\circ)$ will we used to denote a Frobenius manifold, not just a solution to the WDVV equations, and $(g\,,\nabla^\star\,,\star)$ the corresponding almost-dual Frobenius manifold. We suppress, for notational convenience, the full data.},
\[
\begin{array}{cccc}
&(\eta\,,\nabla\,,\circ)& \overset{\partial}{\longrightarrow}&{({\tilde\eta}\,,{\tilde\nabla}\,,\circ)}\\
\text{almost duality}&\downarrow & &\downarrow  \\
&(g\,,\nabla^\star\,,\star) && ({\tilde g}\,,{\tilde \nabla}^\star,\star)\,.
\end{array}
\]
Here $\partial$ is a flat Legendre field, i.e. $\nabla\partial=0\,.$ This ensures we have a map between Frobenius manifolds, rather than just between solutions of the WDVV equations, so the almost-duality transformation may be applied to both sides (the vertical arrows in the above diagram). The aim of this section, and in fact the original aim of this paper, is to construct a transformation to make the above diagram commute, i.e. to construct a Legendre field $\hat\partial$
\[
(g\,,\nabla^\star\,,\star) \overset{\hat{\partial}} {\longrightarrow}({\tilde g}\,,{\tilde \nabla}^\star,\star)\,.
\]
This field - called a twisted Legendre field - was first studied in \cite{RS}. However the full geometric properties were not fully described there.

\begin{thm}
The following diagram commutes:
\[
\begin{array}{ccc}
(\eta\,,\nabla\,,\circ)& \overset{\partial}{\longrightarrow}&{({\tilde\eta}\,,{\tilde\nabla}\,,\circ)}\\
\downarrow & &\downarrow  \\
(g\,,\nabla^\star\,,\star) & \overset{\hat{\partial}} {\longrightarrow}& ({\tilde g}\,,{\tilde \nabla}^\star,\star)\,,
\end{array}
\]
where ${\hat\partial} = E\circ\partial$ and $\nabla \partial=0\,.$ In particular:
\begin{itemize}
\item[(a)] ${\tilde g}(X,Y)= g ( {\hat\partial}\star X, {\hat\partial}\star Y)$\,;
\item[(b)] ${\tilde\nabla}^\star_XY = {\hat\partial}^{-1} \star \nabla_X^\star ( {\hat\partial}\star Y)$\,;
\item[(c)] $\hat\partial$ is a Legendre field for the connection $\nabla^\star\,.$
\end{itemize}

\end{thm}

\begin{proof} We begin by noting that $\partial\circ X = {\hat\partial} \star X\,.$ We also take $\displaystyle{\partial = \frac{\partial~}{\partial t^\kappa}}$ for some $\kappa\,,$ where the $\{t^\alpha\}$ are the flat coordinates of the metric $\eta\,.$

\begin{itemize}
\item[(a)] By definition
\begin{eqnarray*}
{\tilde g}(X,Y) & = & {\tilde \eta}(E^{-1}\circ X,Y) \,,\\
&=& \eta(\partial\circ E^{-1}\circ X,\partial\circ Y)\,,\\
&=&g(\partial\circ X,\partial\circ Y)\,,\\
&=& g({\hat\partial}\star X, {\hat\partial}\star Y)\,.
\end{eqnarray*}

\item[(b)] Starting with the connection at the lower right hand corner and unpacking definitions yields
\begin{eqnarray*}
{\tilde\nabla}^\star_XY &=& E\circ{\tilde\nabla}_X(E^{-1}\circ Y) - ({\tilde\nabla}_{E^{-1}}\circ Y)\circ X + \frac{1}{2} (3-{\tilde d}) E^{-1}\circ X\circ Y\,,\\
&=&E\circ\left\{
\partial^{-1} \circ \nabla_X(\partial \circ E^{-1} \circ Y) \right\}-
\left\{\partial^{-1} \circ \nabla_{E^{-1}\circ Y} (\partial\circ E)\right\}\circ X \\
&&+ \frac{1}{2} (3-{\tilde d}) E^{-1}\circ X\circ Y\,,\\
&=&\partial^{-1}\circ \nabla^\star_X(\partial\circ Y) + \partial^{-1} \circ\left\{ \nabla_{E^{-1}\circ \partial\circ Y}E - \nabla_{E^{-1}\circ Y}(\partial\circ E)\right\}\circ X\\
&&+ \frac{1}{2} (d-{\tilde d}) E^{-1}\circ X\circ Y\,.
\end{eqnarray*}
On using the following formula derived in \cite{DS}
\begin{eqnarray*}
\nabla_{E^{-1}\circ \partial\circ Y}E - \nabla_{E^{-1}\circ Y}(\partial\circ E) &=& \left\{[\partial,E]+[E,e]\circ\partial\right\}\circ Y\circ E^{-1}\,,\\
&=& (d_\kappa-1)\,\partial\circ Y\circ E^{-1}\,
\end{eqnarray*}
(this last line uses standard definitions: $E=\sum (d_i t^i+r_i)\partial_i\,,e=\partial_1\,,$ etc.), hence
\[
{\tilde\nabla}^\star_XY = \partial^{-1} \circ \nabla^\star_X(\partial \circ Y) + \left\{ d_\kappa-1+\frac{1}{2}(d-{\tilde d})\right\}\,E^{-1}\circ X\circ Y\,.
\]
As derived in \cite{Du1}, a Legendre transformation changes the spectrum of the Frobenius manifold, so ${\tilde d} = - 2 \mu_\kappa\,,d=-2\mu_1$ and since $d_\kappa=1-(\mu_\kappa-\mu_1)$ the result follows.

\item[(c)] Finally, as will be proved in the next proposition,
\[
\nabla_X^\star {\hat\partial} = \left( \mu_\kappa+\frac{1}{2} \right) \partial\circ X\,.
\]
With this,
\begin{eqnarray*}
Y \star\nabla_X^\star {\hat\partial} & = & \left( \mu_\kappa+\frac{1}{2} \right) \partial\circ E^{-1} \circ X\circ Y\,,\\
&=&X \star\nabla_Y^\star {\hat\partial}\,.
\end{eqnarray*}
Hence the twisted Legendre field $\hat\partial= E\circ\partial$ is a Legendre field for the connection $\nabla^\star\,.$

\end{itemize}

\end{proof}

It is important to note that, while $\partial$ is flat for $\nabla\,,$ the corresponding twisted Legendre field $\hat\partial$ is not, in general, flat for the connection
$\nabla^\star\,.$ It depends on the spectrum of the original Frobenius manifold.

\begin{prop} The twisted Legendre field $\hat\partial$ satisfies the equation
\[
\nabla_X^\star {\hat\partial} = \left( \mu_\kappa+\frac{1}{2} \right) \partial\circ X\,.
\]
In particular it is flat for $\nabla^\star$ if and only if $\mu_\kappa=-\frac{1}{2}\,.$
\end{prop}
\begin{proof}
By direct calculation
\begin{eqnarray*}
\nabla_X^\star{\hat\partial} & = & E\circ \nabla_X (E^{-1}\circ {\hat\partial} - (\nabla_{E^{-1}\circ {\hat\partial}}E)\circ X + \frac{1}{2}(3-d) E^{-1}\circ {\hat\partial}\circ E^{-1}\,,\\
&=& E \circ \nabla_X \partial - (\nabla_\partial E)\circ X + \frac{1}{2} (3-d) X\circ \partial\,.
\end{eqnarray*}
Since by definition $\delta$ is flat for the torsion free connection $\nabla\,,$
\[
\nabla_X^\star{\hat\partial} = \left\{ [E,\partial] + \frac{1}{2} (3-d) \partial \right\}\circ X\,.
\]
But $[E,\partial]=-d_\kappa\partial$ (since $E=\sum (d_i t^i+r_i)\partial_i$ and $\partial=\partial_\kappa$) and since $d_\kappa=1-(\mu_\kappa-\mu_1)$ and $d=-2\mu_1\,,$ the result follows.
\end{proof}

\begin{rem}
The inverse twisted Legendre field, defined by the equation ${\hat\partial}^{-1} \star {\hat\partial} = E\,,$  is also a twisted Legendre field, in particular,
${\hat\partial}^{-1} = E \circ \partial^{-1}\,.$ From the formulae above, this 
satisfies the equation
\[
{\tilde\nabla}_X^\star {\hat\partial}^{-1} = \hat{\partial}^{-1} \star \nabla_X^\star E\,.
\]
Since, for an almost-dual Frobenius manifold $\nabla_X^\star E\neq 0$ in general, the inverse twisted Legendre field will not be, in general, flat, even if $\hat{\partial}$ is flat for $\nabla^\star\,.$ 
\end{rem}

\begin{ex}\cite{RS}
Starting with the prepotential $F=\frac{1}{2} t_1^2 t_2+e^{t_2}$ we can perform a Legendre transformation with $\partial=\partial_{t_2}$ to obtain a prepotential $\tilde{F}$ together with almost-duality transformations for these two prepotentials to obtain a commuting diagram

\[
\begin{array}{ccc}
F& \overset{\partial}{\longrightarrow}&{\tilde F} \\
\downarrow & &\downarrow  \\
F^\star & \overset{\hat{\partial}} {\longrightarrow}&
{\tilde F}^\star\end{array}
\]
\end{ex}
where
\begin{eqnarray*}
{\tilde F} & = &\frac{1}{2} {\tilde t}_2^2 {\tilde t_1} + {\tilde t}_1^2 \log{\tilde t}_1\,,\\
F^\star & = & \frac{1}{24} z_1^3 - \frac{1}{8} z_1 z_2^2 +\frac{1}{2}\left\{Li_3(e^{z_2}) + Li_3(e^{-z_2})\right\}\,,\\
{\tilde F}^\star & = & \frac{1}{4}\left\{ {\tilde z}_1^2 \log {\tilde z}_1 + {\tilde z}^2_2 \log {\tilde z}_2 - ({\tilde z}_1-{\tilde z}_2)^2 \log({\tilde z}_1-{\tilde z}_2)\right\}
\end{eqnarray*}
and $Li_3(w)$ is the tri-logarithm function.

\noindent We note in passing that the twisted Legendre transformation, in this example, maps trigonometric $\bigvee$-systems \cite{F} to rational $\bigvee$-systems \cite{V}. We end this section by giving a more substantial example.

\begin{ex}
The extended $A_2$ Frobenius manifold found in \cite{DZ} has prepotential
\[
F=\frac{1}{2}t_1^2t_3+\frac{1}{4}t_1t_2^2+t_2e^{t_3}-\frac{1}{96}t_2^4,
\]
and Euler vector field
\[
E=t_1\partial_1+\frac{1}{2}t_2\partial_2+\frac{3}{2}\partial_3.
\]
The flat coordinates, $z_i$, of the intersection form are given, implicitly, by
\begin{eqnarray*}
t_1&=&e^{\frac{2}{3}z_3}(e^{z_1}+e^{-z_2}+e^{z_2-z_1}),\\
t_2&=&e^{\frac{1}{3}z_3}(e^{z_2}+e^{-z_1}+e^{z_1-z_2}),\\
t_3&=&z_3.
\end{eqnarray*}
We can thus use (\ref{symmetry}) to find the almost dual prepotential:
\begin{multline*}
F^*=\frac{1}{2}\left[Li_3(e^{2z_2-z_1})+Li_3(e^{z_1-2z_2})\right]+\frac{1}{2}\left[Li_3(e^{z_2-2z_1})+Li_3(e^{2z_1-z_2})\right]\\ +\frac{1}{2}\left[Li_3(e^{z_1+z_2})+Li_3(e^{-z_1-z_2})\right]+\frac{1}{4}z_1z_2(z_1-z_2)-\frac{2}{3}(z_1^2-z_1z_2+z_2^2) +\frac{2}{27}z_3^3,
\end{multline*}
and can perform generalised Legendre transformations generated by the vector fields
\[
\hat\partial_{t_i}=E\circ\partial_{t_i} \quad i=2,3.
\]
Calculating these vector fields gives
\begin{eqnarray*}
\hat\partial_2&=&\frac{3}{2}e^{t_3}\partial_1+\left(t_1-\frac{t_2^2}{4}\right)\partial_2+\frac{1}{4}t_2\partial_3,\\
\hat\partial_3&=&2t_2e^{t_3}\partial_1+3e^{t_3}\partial_2+t_1\partial_3,
\end{eqnarray*}
or in the z-coordinates
\begin{multline*}
\hat\partial_2=\frac{e^{\frac{z_3}{3}}}{2}\left[\frac{1}{3}\left(-e^{z_1-z_2}+2e^{-z_1}-e^{z_2}\right)\partial_{z_1}\right.
+ \frac{1}{3}\left(e^{z_1-z_2}+e^{-z_1}-2e^{z_2}\right)\partial_{z_2}+\\ \left.\frac{1}{2}\left(e^{z_1-z_2}+e^{-z_1}+e^{z_2}\right)\partial_{z_3}\right]
\end{multline*}
\begin{multline*}
\hat\partial_3=e^{\frac{2z_3}{3}}\left[\frac{1}{3}\left(e^{-z_2}+e^{z_2-z_1}-2e^{z_1}\right)\partial_{z_1}\right.
+ \frac{1}{3}\left(2e^{-z_2}-e^{z_2-z_1}-e^{z_1}\right)\partial_{z_2}+ \\ \left.\left(e^{-z_2}+e^{z_2-z_1}+e^{z_1}\right)\partial_{z_3}\right].
\end{multline*}
We can now perform the transformations as outlined in Proposition \ref{symmetry} for different $\kappa\,.$

\noindent $\bullet$ {Transformation $\hat\partial_2\,:$}
From
\[
\frac{\partial\hat z_\alpha}{\partial z^\beta}=\partial^\gamma c_{\alpha\beta\gamma},
\]
we obtain new coordinates
\begin{eqnarray*}
\hat z_1&=&e^{\frac{1}{3}z_3}(e^{z_1-z_2}-2e^{-z_1}+e^{z_2}),\\
\hat z_2&=&-e^{\frac{1}{3}z_3}(e^{z_1-z_2}+e^{-z_1}-2e^{z_2}),\\
\hat z_3&=&-\frac{3}{2}e^{\frac{1}{3}z_3}(e^{z_1-z_2}+e^{-z_1}+e^{z_2}).
\end{eqnarray*}
which can be inverted
\begin{eqnarray*}
z_1&=&\text{\emph{log}}\left(\frac{(3\hat z_2-2\hat z_3)(3\hat z_1-3\hat z_2-2\hat z_3)}{(3\hat z_1+2\hat z_3)^2}\right)^{\frac{1}{3}},\\
z_2&=&\text{\emph{log}}\left(\frac{-(3\hat z_2-2\hat z_3)^2}{(3\hat z_1+2\hat z_3)(3\hat z_1-3\hat z_2-2\hat z_3)}\right)^{\frac{1}{3}},\\
z_3&=&\text{\emph{log}}\left(-\frac{1}{9^3}(3\hat z_2-2\hat z_3)(3\hat z_1-3\hat z_2-2\hat z_3)(3\hat z_1+2\hat z_3)\right).
\end{eqnarray*}
Then from
\[
\frac{\partial^2\hat F^*}{\partial \hat z_a\partial \hat z_b}=\frac{\partial^2F^*}{\partial  z_a\partial  z_b}
\]
we derive the prepotential
\[
\hat F^*_{(2)}=\sum\limits_{\alpha\in\mathcal{R}}(\alpha\cdot z)^2\text{\emph{log}}(\alpha\cdot z)+ \sum\limits_{\beta\in\mathcal{W}_2}(\beta\cdot z)^2\text{\emph{log}}(\beta\cdot z)
\]
where (after the linear transformation $z_2\rightarrow -z_2$)
\[
    \mathcal{R}=
\begin{cases}
    (1,2,0)& \\
    (1,-1,0)&, \\
    (2,1,0)&
\end{cases}
\]
 the roots of the $A_2$ system and
\[
    \mathcal{W}_2=
\begin{cases}
   \pm (1,0,2/3)& \\
    \pm(0,1,2/3)& .\\
    \pm(-1,-1,2/3)&
\end{cases}
\]

\noindent$\bullet$ {{Transformation $\hat\partial_3\,:$}}
Similarly we find
\begin{eqnarray*}
\hat z_1&=&e^{\frac{2}{3}z_3}(e^{z_2-z_1}-2e^{z_1}+e^{-z_2}),\\
\hat z_2&=&e^{\frac{2}{3}z_3}(e^{z_2-z_1}+e^{z_1}-2e^{-z_2}),\\
\hat z_3&=&-3e^{\frac{2}{3}z_3}(e^{z_2-z_1}+e^{z_1}+e^{-z_2}),
\end{eqnarray*}
inverting,
\begin{eqnarray*}
z_1&=&\text{\emph{log}}\left(\frac{(3\hat z_1-\hat z_3)^2}{(3\hat z_2+2\hat z_3)(3\hat z_1-3\hat z_2+\hat z_3)}\right)^{\frac{1}{3}},\\
z_2&=&\text{\emph{log}}\left(\frac{-(3\hat z_1-\hat z_3)(3\hat z_1-3\hat z_2+\hat z_3)}{(3\hat z_1+\hat z_3)^2}\right)^{\frac{1}{3}},\\
z_3&=&\text{\emph{log}}\left(-\frac{1}{9^3}(3\hat z_1-\hat z_3)(3\hat z_1-3\hat z_2+\hat z_3)(3\hat z_1+\hat z_3)\right)^{\frac{1}{2}}.
\end{eqnarray*}
and the prepotential is
\[
\hat F^*_{(3)}=\sum\limits_{\alpha\in\mathcal{R}}(\alpha\cdot z)^2\text{\emph{log}}(\alpha\cdot z)-2 \sum\limits_{\beta\in\mathcal{W}_3}(\beta\cdot z)^2\text{\emph{log}}(\beta\cdot z)
\]
where $\mathcal{R}$ is as above but now
\[
    \mathcal{W}_3=
\begin{cases}
    \pm(1,0,-1/3)& \\
    \pm(0,1,-1/3)& .\\
    \pm(-1,-1,-1/3)&
\end{cases}
\]

We can also perform the standard Legendre transformations $S_2$ (generated by $\partial_{t^2}$) and $S_3$ (generated by $\partial_{t^3}$) to obtain, respectively, the prepotentials
\begin{eqnarray*}
\hat F_{(2)}&=&\frac{1}{12}\hat t_2^3+\hat t_1\hat t_2\hat t_3+\frac{1}{2}\hat t_1^2\text{\emph{log}}\hat t_1+\frac{1}{3}\hat t_1\hat t_3^3,\\
\hat F_{(3)}&=&\frac{1}{4}\hat t_2^2\hat t_3+\frac{1}{2}\hat t_1\hat t_3^2+\frac{1}{2}\hat t_1^2\text{\emph{log}}\hat t_2.
\end{eqnarray*}
Schematically, then, we have
\[
\begin{array}{ccccc}
\hat F_{(3)}&\overset{S_3}{\longleftarrow}&F& \overset{S_2}{\longrightarrow}&{\hat F_{(2)}}\\
\downarrow & & \downarrow & & \downarrow \\
\hat F^*_{(3)}&\overset{\hat S_3}{\longleftarrow}&F^*& \overset{\hat S_2}{\longrightarrow}&{\hat F^*_{(2)}}
\end{array}
\]
and one may check, by direct calculation, that this diagram commutes.

The geometry of such configurations - the vectors in $\mathcal{R}\cup \mathcal{W}_i$ - are examples of extended $\bigvee$-systems and are studied further in \cite{Sted,SS}.

\end{ex}

\section{The semi-simple case}

It has already been stated that, geometrically, the tensor $\circ$ does not change under a generalized Legendre transformation. What does change are the sets of flat coordinate systems which one uses to find a coordinate representation of this tensor, and hence solutions of the WDVV equations. This invariance actually goes deeper, and this is best seen in the case of a semi-simple multiplication.

At a generic point of a semi-simple manifold, a basis of vector fields, and coordinates $\{u^i\}$, may be found so the multiplication is diagonal,
\[
\frac{\partial~}{\partial u^i} \circ \frac{\partial~}{\partial u^j} = \delta_{ij} \,\frac{\partial~}{\partial u^i}
\]
and with this, the Frobenius condition $\eta(X\circ Y,Z) = \eta(X,Y\circ Z)$ implies that the metric is diagonal, so may be written in the form
\[
\eta\left(
\frac{\partial~}{\partial u^i} , \frac{\partial~}{\partial u^j}
\right) = H_i(u)^2 \delta_{ij}\,,
\]
where the $H_i$ are the Lam\'e coefficients. On writing the generalized Legendre field in the form
\[
\delta = \sum_{i=1}^{dim{\mathcal{M}}} \frac{\psi_i}{H_i} \frac{\partial~}{\partial u^i}
\]
it easily follows that the new metric takes the form
\[
{\tilde\eta}\left(
\frac{\partial~}{\partial u^i} , \frac{\partial~}{\partial u^j}
\right) = \psi_i(u)^2 \delta_{ij}\,.
\]
The Legendre field definition, when written in this coordinate system, reduce to the equation
\[
\frac{\partial~}{\partial u^i} \psi_j = \beta_{ij} \psi_i\,,\qquad i \neq j\,,
\]
where $\beta_{ij}$ are the rotation coefficients (assumed to be symmetric) of the original metric, $\beta_{ij} = \frac{\partial_i H_j}{H_i}\,.$ Re-expressing this shows that the Legendre field condition reduces to the statement that the two diagonal metrics $\eta$ and $\tilde{\eta}$ have the same (off-diagonal) rotation coefficients, and hence flatness will be preserved. But this construction does not imply, as is required in the Frobenius manifold case \cite{Du1}, the condition
$\sum_i\partial_i \psi_j=0$ so one obtains a wider class of flat coordinate systems and solutions of the WDVV equations than in the original construction (where the constraint $\sum_i\partial_i \psi_j=0$ is required to preserve the full Frobenius manifold structure).

\section{Comments}

The Legendre fields - solutions of equation (\ref{Legendre}) - also appear in the theory of hydrodynamics systems associated to $F$-manifolds \cite{LPR}. They appeared there as generators of commuting flows
\[
u_t = \partial \circ u_X
\]
which generalizes the principal hierarchy defined by Dubrovin in the case of Frobenius manifolds \cite{Du1}. The role of $\partial$ in this paper is somewhat different: here it plays the role in defining a symmetry between {\sl different} principal hierarchies. But the use of conservation laws to define new sets of variables for hydrodynamic systems is well established and the theory developed here may be seen as a generalization of this idea (in the sense that the hierarchies as defined above have an underlying connection, and hence one can use $\partial$ to provide a map between such connections). Further, since the Legendre condition comes from the preservation of the torsion-free condition, the theory could be developed to more general situations where one has torsion free, but not metric, connections \cite{AL}.

One unexpected feature of the special case of twisted-Legendre fields (for almost-dual structures) is that they map rational solutions to the WDVV equations to trigonometric solutions, or more specifically, for those almost-dual solutions coming from classical extended-affine-Weyl Frobenius manifolds \cite{DSZZ,SS}.
Since the restriction of Legendre fields to discriminant submanifolds are again Legendre, one should be able to derive similar results for the induced WDVV equations on discriminant submanifolds (for which the induced intersection form is flat), following the ideas in \cite{FV,S}. More generally, one needs to see if one can apply these ideas directly to  rational/trigonometric $\bigvee$-systems.

\section*{Acknowledgements}
Richard Stedman would like to thank the EPSRC for PhD funding (Doctoral Training Grant EP/K503058/1). Ian Strachan would like to thank Liana David for the many conversations and collaborations that led to this paper.


\end{document}